\documentclass[11pt]{article}

\usepackage{etex}
\usepackage[margin={1.2in}]{geometry}
\usepackage[titletoc]{appendix}

\usepackage{tikz}
\usetikzlibrary{matrix}
\usepackage{eucal}

\usepackage{mathrsfs}
\usepackage{stmaryrd}

\usepackage{rotating}

\usepackage{longtable}

\usepackage[sc]{mathpazo}
\usepackage{courier}

\usepackage[utf8]{inputenc}
\usepackage[T1]{fontenc}

\usepackage[english]{babel}
\usepackage{blindtext}

\usepackage{amsmath}

\usepackage{microtype}

\usepackage{amsmath}
\usepackage{amssymb}
\usepackage{amsthm}
\usepackage{color}
\usepackage{latexsym,extarrows}

\usepackage{graphicx}
\graphicspath{ {images/} }
\usepackage{tikz-cd}

\usepackage[all]{xy}

\usepackage[stable,multiple]{footmisc}

\RequirePackage[font=small,format=plain,labelfont=bf,textfont=it]{caption}
\makeatletter
\newsavebox{\@brx}
\newcommand{\llangle}[1][]{\savebox{\@brx}{\(\m@th{#1\langle}\)}%
  \mathopen{\copy\@brx\kern-0.5\wd\@brx\usebox{\@brx}}}
\newcommand{\rrangle}[1][]{\savebox{\@brx}{\(\m@th{#1\rangle}\)}%
  \mathclose{\copy\@brx\kern-0.5\wd\@brx\usebox{\@brx}}}
\makeatother

\begin{document}

\def\e#1\e{\begin{equation}#1\end{equation}}
\def\ea#1\ea{\begin{align}#1\end{align}}
\def\eq#1{{\rm(\ref{#1})}}
\theoremstyle{plain}
\newtheorem{thm}{Theorem}[section]
\newtheorem{lem}[thm]{Lemma}
\newtheorem{prop}[thm]{Proposition}
\newtheorem{cor}[thm]{Corollary}
\theoremstyle{definition}
\newtheorem{dfn}[thm]{Definition}
\newtheorem{ex}[thm]{Example}
\newtheorem{rem}[thm]{Remark}
\newtheorem{conjecture}[thm]{Conjecture}

\newcommand{\mrm}{\mathrm}
\newcommand{\D}{\mathrm{d}}
\newcommand{\A}{\mathcal{A}}
\newcommand{\bL}{\mathbb{L}}
\newcommand{\LL}{\llangle[\Big]}
\newcommand{\RR}{\rrangle[\Big]}
\newcommand{\LD}{\Big\langle}
\newcommand{\bR}{{\mathbb{R}}}
\newcommand{\RD}{\Big\rangle}
\newcommand{\F}{\mathcal{F}}
\newcommand{\HH}{\mathcal{H}}
\newcommand{\X}{\mathcal{X}}
\newcommand{\PP}{\mathbb{P}}
\newcommand{\K}{\mathscr{K}}
\newcommand{\bK}{\mathbb{K}}
\newcommand{\q}{\mathbf{q}}

\newcommand{\op}{\operatorname}
\newcommand{\C}{\mathbb{C}}
\newcommand{\bC}{\mathbb{C}}
\newcommand{\cC}{\mathcal{C}}
\newcommand{\N}{\mathbb{N}}
\newcommand{\R}{\mathbb{R}}
\newcommand{\bT}{\mathbb{T}}
\newcommand{\Q}{\mathbb{Q}}
\newcommand{\Z}{\mathbb{Z}}
\newcommand{\bZ}{\mathbb{Z}}
\newcommand{\NE}{\mathrm{NE}}
\newcommand{\bP}{\mathbb{P}}
\newcommand{\cO}{\mathcal{O}}

\renewcommand{\H}{\mathbf{H}}

\newcommand{\si}{\sigma}
\newcommand{\Si}{\Sigma}

\newcommand{\Etau}{\text{E}_\tau}
\newcommand{\E}{{\mathcal E}}
\newcommand{\G}{\mathbf{G}}
\newcommand{\eps}{\epsilon}
\newcommand{\g}{\mathbf{g}}
\newcommand{\im}{\op{im}}

\newcommand{\h}{\mathbf{h}}

\newcommand{\Gmax}[1]{G_{#1}}
\newcommand{\AW}{E}
\providecommand{\abs}[1]{\left\lvert#1\right\rvert}
\providecommand{\norm}[1]{\left\lVert#1\right\rVert}
\newcommand{\abracket}[1]{\left\langle#1\right\rangle}
\newcommand{\bbracket}[1]{\left[#1\right]}
\newcommand{\fbracket}[1]{\left\{#1\right\}}
\newcommand{\bracket}[1]{\left(#1\right)}
\newcommand{\ket}[1]{|#1\rangle}
\newcommand{\bra}[1]{\langle#1|}

\newcommand{\ora}[1]{\overrightarrow#1}

\providecommand{\from}{\leftarrow}
\newcommand{\bl}{\textbf}
\newcommand{\mbf}{\mathbf}
\newcommand{\mbb}{\mathbb}
\newcommand{\mf}{\mathfrak}
\newcommand{\mc}{\mathcal}
\newcommand{\cinfty}{C^{\infty}}
\newcommand{\pa}{\partial}
\newcommand{\prm}{\prime}
\newcommand{\dbar}{\bar\pa}
\newcommand{\OO}{{\mathcal O}}
\newcommand{\hotimes}{\hat\otimes}
\newcommand{\BV}{Batalin-Vilkovisky }
\newcommand{\CE}{Chevalley-Eilenberg }
\newcommand{\suml}{\sum\limits}
\newcommand{\prodl}{\prod\limits}
\newcommand{\into}{\hookrightarrow}
\newcommand{\Ol}{\mathcal O_{loc}}
\newcommand{\mD}{{\mathcal D}}
\newcommand{\iso}{\cong}
\newcommand{\dpa}[1]{{\pa\over \pa #1}}
\newcommand{\Kahler}{K\"{a}hler }
\newcommand{\0}{\mathbf{0}}

\newcommand{\B}{\mathcal{B}}
\newcommand{\V}{\mathcal{V}}

\newcommand{\M}{\mathfrak{M}}
\newcommand{\fp}{\mathfrak{p}}
\newcommand{\fg}{\mathfrak{g}}

\renewcommand{\Im}{\op{Im}}
\renewcommand{\Re}{\op{Re}}

\newcommand{\DD}{\Omega^{\text{\Romannum{2}}}}

\newtheorem{assumption}[thm]{Assumption}
\newcommand{\tN}{\widetilde{N}}
\newcommand{\tM}{\widetilde{M}}
\newcommand{\bQ}{{\mathbb{Q}}}
\newcommand{\Nef}{\mathrm{Nef}}
\newcommand{\tNef}{{\widetilde{\mathrm{Nef}}}}
\newcommand{\fr}{{\mathfrak{r}}}
\newcommand{\fs}{{\mathfrak{s}}}
\newcommand{\fm}{{\mathfrak{m}}}
\newcommand{\cL}{{\mathcal {L}}}
\newcommand{\tNE}{{\widetilde{\mathrm{NE}}}}
\newcommand{\cX}{{\mathcal{X}}}
\newcommand{\bmu}{{\mathbold{\mu}}}

\newcommand{\ol}{\overline}
\newcommand{\lmap}{\longrightarrow}


\numberwithin{equation}{section}

\makeatletter
\newcommand{\subjclass}[2][2010]{%
  \let\@oldtitle\@title%
  \gdef\@title{\@oldtitle\footnotetext{#1 \emph{Mathematics Subject Classification.} #2}}%
}
\newcommand{\keywords}[1]{%
  \let\@@oldtitle\@title%
  \gdef\@title{\@@oldtitle\footnotetext{\emph{Key words and phrases.} #1.}}%
}
\makeatother

\makeatletter
\let\orig@afterheading\@afterheading
\def\@afterheading{%
   \@afterindenttrue
  \orig@afterheading}
\makeatother

\newcommand{\comment}[1]{\textcolor{red}{[#1]}} 
\title{\bf Modularity of open-closed GW invariants of $K_{\mathbb{P}^2/\mu_3}$}
\author{Yingchun Zhang}\,
\maketitle

\begin{abstract}
Continuing our work in \cite{1805.04894}, this article is devoted to proving that open-closed Gromov-Witten invariants of $K_{\mbb{P}^2/\mu_3}$ are quasi-meromorphic modular forms, and generating functions of open Gromov-Witten invariants are quasi-meromorphic Jacobi forms.
\end{abstract}

\setcounter{tocdepth}{2} \tableofcontents

\section{Introduction}

     The modularity of Gromov-Witten generating function of Calabi-Yau 3-folds is a well-known and yet mysterious phenomenon in Gromov-Witten theory. There is  a great deal of works on this subject \cite{alim2014special,coates2018gromov,zhou2014arithmetic}. We should mention that the "modularity" here should be interpreted in a rather general sense. Among them, an attractive class of examples are toric Calabi-Yau 3-folds where an explicit  B-model is available via Eynard-Orantin topological recursion on mirror curve. This mirror symmetry is remodeling conjecture \cite{eynard2007invariants, bouchard2009remodeling, bouchard2010topological} proved in \cite{fang2016remodeling}. There are 16 local toric surfaces where the mirror curve is genus one. Naturally, one expects its B-model (and hence A-model via remodeling conjecture) to have modularity in the classical sense. 
    
    In \cite{1805.04894}, we have investigated four examples $K_{\mbb{P}^2},\,K_{\mbb{P}^1\times\mbb{P}^1},\,K_{\mbb{P}[1,1,2]},\,K_{\mbb{F}_1}$. Among many things, we showed that its open-closed Gromov-Witten generating functions $F_{g,n}^{\mc{X},\mc{L},f}$ are quasi-meromorphic Jacobi forms up to a mirror map.
   Three of four examples have more than one K$\mrm{{\ddot a}}$hler classes and its generating function is of multi-parameters. Therefore, we need to restrict $F_g$ to some one-dimensional subfamilies. One interesting phenomenon is the multiple-choices of one-dimensional sub-families and the restriction to different one-dimensional subfamily of generating functions may have different modular groups! In this article, we continue our investigation to the example $K_{P^2/\mu_3}$. In particular, we prove
    
\begin{thm}(Theorem 4.5)
\begin{enumerate}
\item 
$d_{\tilde{X}_1}\cdots d_{\tilde {X}_n}\mc{F}_{(g,n)}^{\mc{X},\mc{L}}$ under mirror map is in the ring $\widehat{\mc{J}}(\Gamma)\otimes(H^*(\mc{B}\mu_3,\mbb{C}))^{\otimes n}$, where 
\begin{equation}
\widehat{\mc{J}}(\Gamma)=\widehat{\mc{R}}\left[\{\wp'(u_i-u_r),\wp(u_i-u_r)\}_{i\in\{1,\cdots,n\},r\in R}\right]\,,
\end{equation}
is the differential ring generated by $\wp'(u_i-u_r)$ and $\wp(u_i-u_r)$,
and 
 \begin{equation}
\widehat{\mc{R}}=\mc{M}(\Gamma)\left[\{\wp(u_{r_i}-u_{r_j})\}_{r_i.r_j\in R},\eta_1\right]\,,
\end{equation}
which is the ring of quasi-meromorphic modular forms.
\item
 $\mc{F}_{g}^{\mc{X}}$ is quasi-meromorphic modular form of some group $\Gamma$ under closed mirror map. 
 Group $\Gamma$ depends on the one-dimensional subfamily we choose.
\end{enumerate}
\end{thm}
    
    The main reason for the choice of our first four examples in \cite{1805.04894} is the existence of global hyper-elliptic structure of its mirror curve. Hyper-elliptic structure  is crucial in our argument because its ramification points are 2-torsion points 
    and the coefficients of expansion of $\wp$-function at these points are automatically
   modular forms for $\Gamma(2)$.  The mirror curve of our example now is not globally hyper-elliptic and we do not know if the ramification points are torsion points, but we can still get modularity property similarly. The main reason is ramification points are intersection of 2 curves with coefficients in function field ${\mc{H}/\Gamma}$, and the intersection points then are modular functions of some smaller group $\Gamma'<\Gamma$. In this article, we choose special subfamily to simplify the computation. 
    
    The article is organized as follows. In the section 2, we will briefly review construction of mirror curve and remodeling conjecture in this case. The appropriate one-dimensional subfamily is constructed in the section 3. The main results are proved 
    in the section 4. In many ways, this article is a continuation of \cite{1805.04894}. We were informed that a general treatment for remaining local toric surfaces are prepared in a up-coming work of Fang-Liu-Zong.

\section*{Acknowledge}
I would like to express my thanks to Prof. Bohan Fang for discussion about remodeling conjecture, to Prof. Jie Zhou for his great patience on B model discussion and for many good suggestions and ideas, and also to Prof. Yongbin Ruan for his support and useful discussion. This article is partially supported by the China Scholarship Council.
\section{Revisit Remodeling Conjecture}
\subsection{Mirror curve construction}
The main reference of this subsection is \cite{fang2016remodeling}. Let $\mc{X}=K_{\mbb{P}^2/\mu_3}$, defined by a fan $\Sigma\subseteq N_{\mbb{R}}$. $N=\mbb{Z}^3$ is a lattice of rank 3. $\Sigma(1)=\{\rho_1,\cdots,\rho_{p'+3}\}$ is set of 1-dimensional cone of $\Sigma$, and $b_i\in \rho_i\cap \mbb{Z}$ is integral generator of $\rho_i$. Since $\mc{X}$ is Calabi-Yau, all $b_i\,,i=1,\cdots,p'+3$ lie in a hyperplane of $N$, and assume it to be $N':=\mbb{Z}^2\times\{1\}\subseteq N$, so $b_i=(m_i,n_i,1)$. After change of basis of $N'$, assume that $b_1=(r,-s,1),\,b_2=(0,m,1),\,b_3=(0,0,1)$. Let $\Delta=\Sigma\cap N'_{\mbb{R}}$ be a polytope and $\Delta\cap N'={(m_i,n_i),i=1,\cdots,p+3}$, where $p+3=|\Delta\cap N'|$. $\Delta$ is defining polytope of $\mc{X}$. $\tilde{b}_i=(m_i,n_i,1)\,, i=1,\cdots, p+3$ and $\tilde b_i=b_i$ for $i=1,\cdots,p'+3$. Let $\si_0$ is a 3-dimensional cone with $b_1,\,b_2,\,b_3$ as edges, and $\tau_0$ is a 2-dimentional subcone with $b_2,\,b_3$ as edges. As is shown in below picture $\si_0$ is the cone over shadowed triangle and $\tau_0$ is the 2 dimensional cone over the vertical edge of $\si_0$.

Trianglized defining polytope of $K_{P^2/\mu_3}\,$,\\
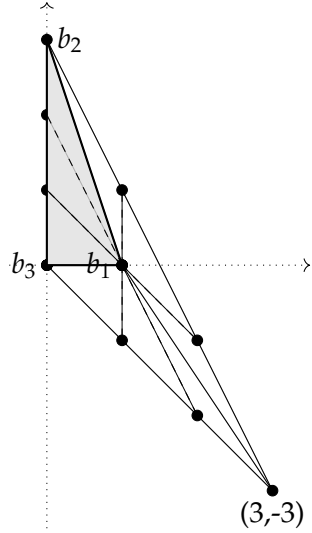
\begin{center}
\begin{tikzpicture}{definingpolytopeofmirrorP2}
\draw [dotted, ->] (-0.5,0) -- (3.5,0); 
\draw [dotted, ->] (0,-3.5) -- (0,3.5); 
\filldraw (0,0)circle(2pt) node[left]{$b_3$}--(1,-1)circle(2pt)--(2,-2) circle(2pt)--(3,-3)circle(2pt) node[below]{(3,-3)}--(2,-1)circle(2pt)--(1,1)circle(2pt)--(0,3)circle(2pt) node[right]{$b_2$}--(0,2)circle(2pt)--(0,1) circle(2pt)--(0,0)circle(2pt);
\draw[thick, fill=black!10] (1,0) -- (0,3) -- (0,0) -- cycle;
\filldraw (1,0)circle(2pt)--(0,0);
\filldraw(1,0)circle(2pt)--(3,-3);
\filldraw (1,0)circle(2pt)--(0,3);
\filldraw [dashed](1,0)circle(2pt) node[left]{$b_1$}--(0,2);
\filldraw [dashed](1,0)--(0,1);
\filldraw [dashed](1,0)--(1,-1);
\filldraw [dashed](1,0)--(2,-2);
\filldraw [dashed](1,0)--(2,-1);
\filldraw [dashed](1,0)--(1,1);
\end{tikzpicture}
\end{center}
Let
\begin{equation}
\widetilde{N}=\bigoplus_{i=1}^{p+3}\mbb{Z}\tilde{b}_i
\end{equation}
and there is a exact sequence of group homomorphism,
\begin{equation}\label{exactsequence}
0\longrightarrow \mbb{L}\overset{\phi}\longrightarrow \widetilde{N}\overset{\psi}{\longrightarrow} N\longrightarrow 0\,,
\end{equation}
where $\mbb{L}\cong \mbb{Z}^p$.
Applying tensor product by $\mbb{C}^*$ to \eqref{exactsequence},
\begin{equation}
1\longrightarrow G\longrightarrow \widetilde{T}{\longrightarrow} T\longrightarrow 1\,.
\end{equation}
Action of $\widetilde{T}$ on itself extends to $\mbb{C}^{p+3}=Spec[Z_1,\cdots,Z_{p+3}]$. Let $Z(\Sigma)=\bigcup_{\sigma\in \Sigma(3)}\{Z_i=0,\, \mrm{if}\,\, \rho_i\nsubseteq \sigma\}$, $U(\Sigma)=\mbb{C}^{p+3}-Z(\Sigma)$. Then
\begin{equation}
\mc{X}=[U(\Sigma)/G]\,.
\end{equation}
Applying $\mrm{Hom}(-,\mbb{Z})$ to \eqref{exactsequence}, we have a long exact sequence, 
\begin{equation}
0\longrightarrow M\overset{\psi^\vee}\longrightarrow \widetilde{M} \overset{\phi^\vee}\longrightarrow \mbb{L}^\vee\longrightarrow 0\,.
\end{equation}
Assume $\{\tilde{b}_i^{\vee},i=1,\cdots,\,p+3\}$ is a set of basis of $\widetilde{M}$
and $D_i=\phi^{\vee}(\tilde{b}_i^{\vee})$.
Denote 
\begin{equation}
I'_{\sigma}=\{i\in \{1,\cdots,p'+3\},\,st.\,\rho_i\in \sigma\}\,, 
I_{\sigma}=\{1,\cdots,p+3\}\setminus I'_{\sigma}\,.
\end{equation}
Extended nef-$\si$ cone is defined as, 
\begin{equation}
\widetilde{Nef}_{\si}:=\sum_{i \in I_\si}\mbb{R}_{\geqq 0}D_i\,,
\end{equation}
and extended nef cone of $\mc{X}$ is
\begin{equation}
\widetilde{Nef}(\Si):=\bigcap_{\si\in \Si(3)}\widetilde{Nef}_{\si}\,.
\end{equation}
Take a subset $\{H_i,\,i=1\cdots,p\}$ of $\widetilde{Nef}\cap \mbb{L}^{\vee}$ which is a basis of $\mbb{L}_{\mbb Q}^{\vee}$, and require that image of $\{H_i,\,i=1\cdots,p'\}$ in $H^2(\mc{X},\mbb{Q})$ forms a set of basis. Let $H_i=D_{i+3},\, \mrm{for}\, i=p'+1,\cdots, p$.
 
Assume
\begin{equation}
H_a=\sum_{i\in I_{\si_0}}s_{aj}D_j\,,
\end{equation}
and define a monomial of $q=(q_1,\cdots,q_p)$

\begin{align}
a_i=\begin{cases}
1,&i\in I'_{\si_0}\\
\prod_{a=1}^pq_a^{s_{ai}},&i\in I_\si
\end{cases}\,.
\end{align}
Then the miror curve, denoted by $\mc{C}_q$, has equation 
\begin{equation}\label{toriccoordmirrorcurve}
X+Y^3+1+a_1X^3Y^{-3}+a_2XY^{-1}+a_3X^2Y^{-2}+a_4X^2Y^{-1}+a_5XY+a_6Y^2+a_7Y=0\,,
\end{equation}
with
\begin{align}
a_1=q_1^3, a_2=q_1q_2, a_3=a_1^2q_3, a_4=q_1^2q_4, a_5=q_1q_5, a_6=q_6, a_7=q_7.
\end{align}
The compactification of mirror curve $\overline{\mc{C}}_q$ is the natural compactification when embedding $\mc{C}_q \hookrightarrow \mbb{P}_{\Delta}$. Let $\mc{U}$ be moduli space of complex structure of $\mc{C}_q$ over which $\mc{C}_q$ and $\overline{\mc{C}}_q$ are smooth. 
\subsection{Aganagic-Vafa brane}
Let $T'=N'\otimes \mbb{C}^*$ and $T'_{\mbb{R}}$ be maximal compact subgroup of $T'$. From the construction of toric Calabi-Yau orbifold, action of $T'_{\mbb R}$ preserves the Calabi-Yau structure of $\mc{X}$. The corresponding moment map of this Hamiltonian action is
\begin{equation}
\mu_{T'_{\mbb R}}:\,\mc{X}\mapsto M'_{\mbb R}\,.
\end{equation}
Let $\mc{X}^1$ be the union of $T'_{\mbb{R}}$ invariant 0-dimensional and 1-dimensional orbits. The image $\mu_{T'_{\mbb R}}(\mc{X}^1)$ is called toric graph. Choose a point in one outer leg of toric graph, then $\mc{L}\subset \mu_{T'_{\mbb R}}^{-1}(p)$, with an additional condition that when we write $\mc{X}$ as GIT quotient $\mbb{C}^{p+3}\sslash(\mbb{C}^*)^p$, then sum of angles of $Z_i$ on $\mc{L}$ is constant.  $\mc{L}$ is called outer Aganagic-Vafa brane. Assume the 2-dimensional cone corresponding to the outer leg is the one considered in last section $\tau_0$. In our example, $\mc{L}\cong \left[\mbb{C}\times S^1/\mu_3\right]$.

\subsection{Open-closed Gromov-Witten invariants}
Denote $\beta\in H_2(\mc{X},\,\mc{L};\,\mbb{Z})$, and $\vec{\mu}=(\bar{\mu}_1,\cdots,\bar{\mu}_n)\in H_1(\mc{L},\mbb{Z})^n$, In our example, $H_1(\mc{L},\mbb{Z})\cong \mbb{Z}\times\mu_3$, so $\bar{\mu}_i=(d_i,k_i)$. 
\begin{eqnarray}
\overline{\mc{M}}_{(g,h),n}(\mc{X},\mc{L},|\,\beta,\vec{\mu}):=\Big\{(\Si,\partial \Si)\overset{\phi}\longrightarrow (\mc{X},\mc{L})\,|\, \phi\, \mrm{stable}, \phi_{*}([\Si])=\beta, \phi_{*}([\partial(\Si)])=\vec{\mu} \Big\}\,,
\end{eqnarray}
is moduli space of stable maps from a Rieman surface of $g$ holes $n$ boundaries and $h$ markings to $(\mc{X},\mc{L})$.
\begin{equation}
ev_i: \overline{\mc{M}}_{(g,h),n}(\mc{X},\mc{L},|\,\beta,\vec{\mu}) \mapsto \mc{IX}\,,
\end{equation}
is evaluation map.
$\gamma_1,\cdots,\gamma_h\in H_{T',\,CR}^*(\mc{X},\mbb{Q})$, then open Gromov-Witten invariant is defined to be 
\begin{equation}\label{gwinvariants}
<\gamma_1,\cdots,\gamma_h>_{(g,h),n}^{\mc{X},\mc{L},{T'_{\mbb{R}}}}:=\int_{[F]^{vir}}{\prod_{i}ev_i^*(\gamma_i)\over N_F^{vir}}\,,
\end{equation}
where $F$ is the toric $T'_{\mbb{R}}$ fixed loci. 
Then \eqref{gwinvariants} lies in $\mbb{Q}(w_1,w_2)$, fractional field of $H^*(\mc{B}T'_{\mbb{R}},\mbb{Q})$. Set $w_1=1,\,w_2=f$, then \eqref{gwinvariants} is $\mbb{Q}$-valued and denoted by $<\gamma_1,\cdots,\gamma_h>_{(g,h),n}^{\mc{X},(\mc{L},f)}$. Let $\textbf{T}=\sum_{i=1}^p t_iT_i$, and $\{T_i\}_{i=1}^p$ is a set of basis of $H^2_{CR}(\mc{X},\mbb{Q})$.
Open-closed potential is defined in \cite{fang2012open} as 
\begin{align}
\mc{F}_{g,n}^{\mc{X},(\mc{L},f)}\left(\textbf{T},\tilde{X}_1,\cdots,\tilde{X}_n\right)=\sum_{d\in Eff(\mc {X})}\sum_{d_1,\cdots,d_n>0}\sum_{k_1,\cdots,k_n=0}^{2}&\sum_{h=0}^{\infty}{{<}\textbf{T}^h{>}_{(g,h),n}^{\mc{X},(\mc{L},f)}\over h!}\bullet \nonumber \\ &\prod_{j=1}^{n}\tilde{X}_j^{n_j}\bigotimes_{i=1}^n\left(-(-1)^{-k_1\over 3}\right)\textbf{1}'_{-k_i\over 3}
\end{align}
It takes value in $H^*_{CR}\left(\mc{B}\mu_m,\mbb{C}\right)$. When $n=0$, it recovers closed Gromov-Witten potential which is denoted by $\mc{F}_g^{\mc{X}}$.

\subsection{Statement of remodeling conjecture}
Remodeling conjecture, proposed in \cite{bouchard2009remodeling, bouchard2010topological}, relates A-model open-closed Gromov-Witten invariants with B-model Eynar-Orantin topological recursion invariants under mirror map proved in \cite{fang2016remodeling}.

Consider the mirror curve constructed before. Assume $p:\widetilde{\mc{C}}_q\longrightarrow \mc{C}_q$, is the $\mbb{Z}^2$ cover map. Then closed mirror map is given by 
\begin{eqnarray}
t_i=\int_{\tilde{A}_i}logY{d\hat{X}\over \hat X},\,i=1,\cdots,p\,,
\end{eqnarray}
and 
\begin{equation}
\bar A_i=p_*(\tilde A_i),\,i=1,\cdots,p\,
\end{equation}
lie in $K(\mc{C}_q,\mbb{Z})$, where
 \begin{equation}
 K(\mc{C}_q)=\mrm{kernel}(H_1(\mc{C}_q,\mbb{Z})\mapsto H_1((\mbb{C}^*)^2),\mbb{Z})\,.
 \end{equation}
$\{A_i\}_{i=1}^g$ is the image of $\{\bar A_i\}$ in $H_1(\overline{\mc{C}}_q,\mbb{Z})$.
Choose $\{B_i\}_{i=1}^g\subset H_1(\overline{\mc C}_q,\mbb{Z})$, st.
\begin{equation}
A_i\cap B_j=\delta_{ij}\,,\, \,A_i\cap A_j=B_i\cap B_j=0\,.
\end{equation}

Bergman kernel $B(p,q)$ normalized by $\{A_i\}$ is a meromorphic symmetric bidifferential 2-form on $\overline{\mc{C}}_q\times\overline{\mc{C}}_q$ characterized by
\begin{itemize}
\item B(p,q) is holomorphic except that it admits 2 order of pole on diagnal.
\begin{equation}
B(p,q)=\left({1\over (p-q)^2}+f(p,q)\right)dpdq\,,
\end{equation}
in which $f(p,q)$ is symmetric.
\item
\begin{equation}
\int_{q\in A_i}B(p,q)=0,\,i=1,\cdots,g.
\end{equation}
\end{itemize}
Choose $\omega_i\in H^0(\overline{\mc C}_q,\omega_{\overline{\mc C}_q})$, satisfying
\begin{equation}
\int_{A_j}\omega_i=\delta_{ij}, \int_{B_j}\omega_i=\tau_{ij}\,.
\end{equation}
In our example, the compactified mirror curve is genus 1, so $g=1$. 
Choose a base point $o$ and we have Abel-Jacobi map,
\begin{eqnarray}
u: \overline{\mc{C}}_q &\mapsto &\mc{J}(\overline{\mc{C}}_q)=\mbb{C}/{\mbb{Z}\oplus\tau\mbb{Z}}\,,\nonumber\\
    p & \mapsto & u(p)=\int_{o}^p\omega \,.
\end{eqnarray}
Then the Bergman kernel has an expression in Jacobi coordinate 
\begin{equation}
B(p,q)=\left(\wp(u(p)-u(q),\tau)+\eta_1\right)du(p)du(q)\,,
\end{equation}
with $\eta_1={\pi^2\over 3}E_2$.

We use Schiffer kernel $S(p,q)$ as anti-holomorphic completion of Bergman kernel, such that it is independent the choice of Torelli marking and can be extended to the whole moduli space of complex structure. The method to do antiholomorphic completion is explained in \cite{eynard2007invariants}. 
\begin{equation}
S(p,q)=\left(\wp(u(p)-u(q),\tau)+\hat{\eta}_1\right)du(p)du(q)\,,
\end{equation}
where $\hat{\eta}_1=\eta_1-{\pi \over Im\tau}$, and 
\begin{equation}
\lim_{Im\tau\to \infty}S(p,q)=B(p,q)\,.
\end{equation}

Let $R=\{p\in \mc{C}_q, d(XY^f)=0\}$ be set of ramification points in mirror curve. 
$\omega_{0,1}=0,\,\,\widehat{\omega}_{0,2}=S(p,q)$, then for $2g-2+n>0$,

\begin{align}\label{eorecursionrelation}
\widehat{\omega}_{g,n}^f(p_1,p_2\cdots p_n)=
&\sum_{r\in R}Res_{q\to r}{{1\over 2}\int_{q^*}^q S(p_1,\xi)\over \lambda-\lambda^*}
\Bigg\{ \widehat{\omega}_{g-1,n+1}^f(q,q^*,p_1\cdots,p_n)\nonumber\\
&+\sum_{I\cup J=\{2,\cdots,n\}}\sum_{g_1+g_2=g}\widehat{\omega}_{g_1,|I|+1}^f(q,p_I)\widehat{\omega}_{g_2,|J|+1}^f(q^*,p_J)
\Bigg\}\,.
\end{align}

$\hat{X}:=XY^f: \overline{\mc{C}}_q \mapsto \mbb{C}$. $\nu_l\in\hat{X}^{-1}(0),\,l=0,\,1,\,2$ in boundary of $\mc{C}_q$ are open GW points, and assume $Y_l=Y(\nu_l)$ . In a neighborhood $\nu_l\in D_q^l\subset \ol{\mc{C}}_q$, assume $\hat{X}(D_q^l)\subset D_{\delta}=\{\hat{X}\in \mbb{C}:\hat{X}\leq \delta\}$.
Then there are maps 
\begin{equation}
\rho_q^{l_1,\cdots,l_n}:=(\hat{X}_q^{l_1})^{-1}\times\cdots\times(\hat{X}_q^{l_n})^{-1}: (D_{\delta})^n\mapsto D_q^{l_1}\times\cdots\times D_q^{l_n}\,.
\end{equation}
$(\rho_q^{l_1,\cdots,l_n})^*\widehat{\omega}_{g,n}^f(p_1,\cdots,p_n)$ means expanding $\widehat{\omega}_{g,n}^f$ near the open Gromov point $({\nu}_{l_1},\cdots,\nu_{l_n})$.
\\
Assume the corresponding $u$ coordinate of the open GW points are $w_l,\,l=1,2,3$, and  $w_l\in\widetilde{D}^l$ is a neighborhood of $w_l$ in $\mc{J}(\overline{\mc{C}}_q)$. and 
\begin{equation}
\delta^{l_1,\cdots,l_n}:=(u^{l_1})^{-1}\times\cdots\times(u^{l_n})^{-1}: (D_{\delta})^n\mapsto \widetilde{D}^{l_1}\times\cdots\times \widetilde{D}^{l_n}\,.
\end{equation}

Let $H_{CR}^*(\mc{B}\mu_3,\mc{C})=\mc{C}\textbf{1}\oplus\mc{C}\textbf{1}_{1\over 3}\oplus\mc{C}\textbf{1}_{2\over 3}$, and $\psi_l={1\over 3}\sum_{k=0}^{2}\xi_{3}^{-kl}\textbf{1}'_{k\over 3}$, $l=0,1,2\,$.

Take limit $Im\tau \to \infty$, then we recover $\omega_{g,n}^{f}$.
The remodeling conjecture applied to modularity purpose says that
\begin{itemize}
\item For $2g-2+n>0$,
\begin{equation}\label{mirrortheorem}
d_{\tilde{X}_1}\cdots d_{\tilde {X}_n}\mc{F}_{(g,n)}^{\mc{X},(\mc{L},f)}={(-1)^{g-1+n}\over 3^{n}}\sum_{l_1,\cdots,l_n \in \{1,2,3 \}}(\delta^{l_1\cdots l_n})^*\omega_{g,n}^f(u_1,\cdots,u_n)\psi_{l_1}\otimes\cdots\otimes\psi_{l_n}\,,
\end{equation}
\item
\begin{equation}
d_{\tilde{X}_1}d_{\tilde {X}_2}\mc{F}_{(0,2)}^{\mc{X},(\mc{L},f)}=-3^{-2}\sum_{l_1,l_2\in \{1,2,3\}}\left((\delta^{l_1l_2})^*\omega_{(0,2)}-{du_1du_2\over (u_1-u_2)^2}\right)\psi_{l_1}\psi_{l_2}\,,
\end{equation}
\item
\begin{equation}
d_{\tilde X}\mc{F}_{(0,1)}^{\mc{X},(\mc L,f)}={1\over 3}\sum_{l\in \{1,2,3\}}\left(\ln Y(\rho_l{\hat X})-\ln Y_l\right){d\hat X\over \hat X}\psi_l\,,
\end{equation}
\item 
\begin{equation}
\mc{F}_g^{\mc{X}}=\widehat{\mc{F}}_{g}:={1\over 2g-2}\sum_{r\in R}Res_{q\to r}d^{-1}\lambda\cdot\omega_{g,1}\,.
\end{equation}
\end{itemize}

The $"="$ means that under mirror map and right-hand side expanded near open GW points, two sides are equal.

\section{A special subfamily of mirror curve} 
Although compared with the example in \cite{1805.04894}, mirror curve of $K_{\mbb{P}^2/\mu_3}$ has no hyperelliptic structure, and the ramification points are not necessary to be torsion points, we can still get similar results. The reason is that in Eynard-Orantin recursion, the modularity of ramification points plays a crucial important role, and in this example, we can still get the modularity property of ramification points. 

In this section, I will focus on a concrete and easy subfamily to simplify computation and highlight the idea. 

Take 1-dimensional subfamily by letting
\begin{equation}
q_{i}=0\,,\,\mrm{for} \,\,i=2,\cdots,7\,,
\end{equation}
and $q={1\over q_1}$, and then the mirror curve subfamily becomes
\begin{align}\label{subfamilycurve}
H(X,Y):=X^3+Y^3+Y^6+qXY^3=0\,.
\end{align}
Denote this subfamily by $\mc{C}$.
In this subfamily, it is enough to choose framing $f=0$. Denote $X=\hat X$ that appears in section 2.

\subsection{Primary definition of modular form}
In this subsection, I will introduce some basic and necessary definition of modular forms. See \cite{diamond2005first}.
\begin{dfn} A holomorphic function $f:\mc{H}\mapsto \mbb{C}$ is a modular form of weight $k\in \mbb{Z}$ for group $\Gamma\subset SL(2,\mbb{Z})$ if\\
 (1) for any $\begin{pmatrix}a&b\\c&d \end{pmatrix}\in \Gamma$, 
\begin{equation}
f\left({a\tau+b\over c\tau+d}\right)=(c\tau+d)^kf(\tau)\,.
\end{equation}
 (2)
$f(\tau)$ is holomorphic at $\infty$ .
Denote the ring to be $M_k(\Gamma)$.
\end{dfn}
\begin{ex}
Weight 4 and 6 Eisenstein modular form 
\begin{eqnarray}
&&E_4={1\over 4\zeta(4)}{\sum_{c\in \mbb{Z}}\sum_{d\in \mbb{Z}}}'{1\over (c\tau+d)^4}\,,\nonumber\\
&&E_6={1\over 6\zeta(6)}{\sum_{c\in \mbb{Z}}\sum_{d\in \mbb{Z}}}'{1\over (c\tau+d)^6}\,.
\end{eqnarray}
\end{ex}

\begin{dfn} A holomorphic function $f:\mc{H}\longrightarrow \mbb{C}$ is a quasi-modular form if\\
(1)
\begin{equation}
f\left({a\tau+b\over c\tau+d}\right)=(c\tau+d)^kf(\tau)+\sum_{i=0}^{k-1}c^k(c\tau+d)^{k-1}f_j\,,
\end{equation}
$f_j$ is holomorphic function.\\
(2) 
$f$ is holomorphic at $\infty$. Denote this ring to be $\widehat M_k(\Gamma)$.
\end{dfn}

\begin{ex}
An example of quasi-modular form is weight 2 Eisenstein series.
\begin{equation}
E_2={1\over 2\zeta(2)}{\sum_{c\in \mbb{Z}}\sum_{d\in \mbb{Z}}}'{1\over (c\tau+d)^2}\,.
\end{equation}
\end{ex}

\begin{dfn} Almost holomophic modular form of weight k is an antiholomorphic function $f:\mc{H}\longrightarrow \mbb{C}$ satisfying,\\
(1) $f$ behaves like modular forms under group action
\begin{equation}
f\left({a\tau+b\over c\tau+d}\right)=(c\tau+d)^kf(\tau)\,.
\end{equation}
(2) $f(\tau)$ can be expanded as polynomial of ${1\over Im\tau}$, with holomorphic function as coefficients. Denote the ring of almost holomorphic modular forms by $\widetilde M_k(\Gamma)$. 
\end{dfn}
For example $\widehat{E}_2=E_2-{3\over \pi Im\tau}$ is an almost holomorphic modular form.

Let $\mc{M}(\Gamma)$ denote the ring of meromorphic modular forms which is fractional field of $M(\Gamma)$, and then we have ring of quasi-meromorphic modular forms $\widehat{\mc{M}}(\Gamma)=\mc{M}(\Gamma)[E_2]$ and ring of almost meromorphic modular forms $\widetilde{\mc{M}}(\Gamma)=\mc{M}[\hat E_2]$.

In this article, we will deal with modular forms with non-trial multiplier system because of $\eta$ function which is an modular form with non-trivial multiplier system, so modular form in this article means modular form with multiplier system.

\begin{dfn}Jacobi form of weight k index m of a group $\Gamma$ is a holomorphic  function $\phi:\mbb{C}\times\mc{H}\longrightarrow \mbb{C}$ satisfying\\
(1)
\begin{equation}
\phi\left({z\over c\tau+d},{a\tau+b\over c\tau+d}\right)=e^{2\pi i m c z^2\over c\tau+d}(c\tau+d)^k\phi(z,\tau)\,.
\end{equation}
(2) For $\lambda\,,\mu\in \mbb{Z}$,
\begin{equation}
\phi(z+\lambda\tau+\mu,\tau)=e^{-2\pi i m(\lambda^2\tau+2\lambda Z)}\phi(z,\tau)\,.
\end{equation}
(3) $\phi$ has Fourier expansion
\begin{equation}
\phi(z,\tau)=\sum_{n,r}C(n,r)e^{2\pi i(n\tau+rz)}\,,
\end{equation}
$C(n,r)=0$ unless $4mn\geq r^2$.
\end{dfn}
$J(\Gamma)$ denotes the ring of Jacobi forms of group $\Gamma$.
Let $\mc{J}(\Gamma)$ denote the ring of meromorphic Jacobi forms which is the fractional field of $J(\Gamma)$.

\begin{ex}
A useful meromorphic Jacobi form of index 0 weight 2 is Weierstrass $\wp$ function,
\begin{equation}
\wp(z,\tau)={1\over z^2}+\sum_{(m,n)\in \mbb{Z}^2\setminus(0,0)}\left({1\over (z+m\tau+n)^2}-{1\over( m\tau+n)^2}\right)\,.
\end{equation}
$\wp'(z)$ is meromorphic Jacobi form for $SL(2,\mbb{Z})$ of weight 3.
\end{ex}

Quasi-meromorphic Jacobi forms ring $\widehat{\mc J}(\Gamma):=\mc J(\Gamma)\otimes\widehat{M}(\Gamma)$, and almost meromorphic holomorphic modular form is $\widetilde{\mc J}(\Gamma)=\mc J(\Gamma)\otimes\widetilde{M}(\Gamma)\,$.

\subsection{$\wp$-Uniformization of mirror curve}
Every elliptic curve is birationally equivalent to Weierstrass normal form by calssical theory of elliptic curves, and there is a simple procedure called Nagell’s algorithm to carry the birational transformition\cite{connell1996elliptic}. 
Fix a Torelli marking $\{A,B\}$ of $\overline{\mc{C}}_q$ and a holomorphic 1-form $\omega$, such that
\begin{equation}
\int_A\omega=1\,,\,\,\int_{b}\omega=\tau \,.
\end{equation}
Fixing a base point $o$, we have Abel-Jacobi map
\begin{eqnarray}
u:\mc{C}&\mapsto&\mc{J}{(\mc{C})}=\mbb{C}/\mbb{Z}\oplus\mbb{Z}\tau \,,\nonumber\\
p & \mapsto &u(p)=\int_{o}^p \omega
\end{eqnarray}
In order to do $\wp$-uniformization, change coordinate firstly
\begin{equation}
x=XY^{-1},\, y=Y\,.
\end{equation}\label{mirrorcurvetrans}
Then mirror curve is transformed to 
\begin{equation}\label{subfamilycurve2}
x^3+1+y^3+qxy=0\,.
\end{equation}
which is a Hesse pencil curve. 
Uniformization of \eqref{subfamilycurve2} is
\begin{align}
&t={-12-{1\over 9}q^3+{4^{1/3}\over 3}q\kappa^{-2}U\over 4^{1\over 3}\kappa^{-2}U+q^2}\,,\nonumber\\
&x=(3t-q){-4(27+q^3)t+\kappa^{-3}V(3t-q)\over -8(27+q^3)(1+t^2)}\,,\nonumber\\
&y=tx-1 \,,
\end{align}
with inverse map
\begin{eqnarray}\label{inverseuniformization}
&&U={-\kappa^2\over 3}{108x+{q^3}x+9q^2y+9q^2\over 4^{1\over 3}(3y+3-qx)}\,,\nonumber\\
&&V=-4(q^3+27)\kappa^3{3y^2-qxy+qx-3\over(3y+3-qx)^2}\,.
\end{eqnarray}
Under birational transformation, \eqref{subfamilycurve2} becomes 
\begin{equation}
V^2=4U^3-g_2U-g_3\,.
\end{equation}
with 
\begin{equation}
U=\wp(u,\tau)\,,V=\wp'(u,\tau)\,,
\end{equation}
and
\begin{eqnarray}
&g_2&={ {\kappa^44^{1\over 3}\over 3}(q^4-216q)}\,,\nonumber \\
&g_3&={-{2\over 27}q^6-40q^3+432}.
\end{eqnarray}
The j-invariant is 
\begin{equation}
j=-{q^3 (-216 + q^3)^3\over (27 + q^3)^3}\,.
\end{equation}
We can set 
\begin{equation}\label{qformula}
q(\tau)=-\left(3^3+{\eta(\tau)^{12}\over \eta(3\tau)^{12}}\right)^{1\over 3}\,.
\end{equation}
Then $q(\tau)^3$ is Hauptmodul of $\Gamma_0(3)$, so compactified mirror curve \eqref{subfamilycurve2} is a curve family over $\mc{H}/\Gamma_0(3)$.  
$q(\tau)$ is also equal to
$q(\tau)=-3-\left({3\eta(3\tau)\over \eta({\tau\over 3})}\right)^3$, and it is a Hauptmodul of $\Gamma(3)$.
Let
\begin{equation}\label{g2g3}
g_2=60\cdot 2\cdot\zeta(4)\cdot E_4(\tau)\,,\,
g_3=140\cdot 2\cdot\zeta(6)\cdot E_6(\tau)\,,
\end{equation}
then
\begin{equation}\label{kappa}
\kappa=\left({3\cdot{4\over 3}\pi^4 E_4\over {4}^{1\over 3}q(q^3-216)}\right)^{1\over 4}\,,
\end{equation}
and $\kappa$ is a modualr form of weight 1 with some multiplier system.
By coordinate transformition \eqref{mirrorcurvetrans}, we get the uniformization of mirror curve \eqref{subfamilycurve}.

\subsection {Ramification point}
Let $R=\{r\in \mc{C}_r, d|_p(X)=0\}$ be set of ramification points, they are not necessary to be ramification points, but in the subfamily we choose now, ramification points still have modularity property. 
Since $r\in R$ is the intersection of the following two plane curves
\begin{align}\label{criticaldef}
&X^3+Y^3+Y^6+qXY^3=0\,,\nonumber\\
&1+2Y^3+qX=0\,.
\end{align}
which give us 9 ramification points.
Explicitly, $Y^3$ has 3 different solutions. 
\begin{eqnarray}\label{ramificationpts}
&&Y_1^3={1\over 24}(-12-q^3)+{-24 q^3 - q^6\over 24 t} - {1\over 24} t\,,\nonumber
 \\
&&Y_2^3={1\over 24}(-12-q^3)-{(1+i\sqrt 3)(-24q^3-q^6)\over 48t}+{1\over 48}(1 - i \sqrt 3)t\,,\nonumber
\\
&&Y_3^3={1\over 24}(-12-q^3)-{(1 - i \sqrt 3) (-24 q^3 - q^6)\over
    48 t} + {1\over 48} (1 + i\sqrt 3)t\,,
\end{eqnarray}
with
\begin{equation}\label{criticalpoints}
t=\left(216 q^3 + 36 q^6 + q^9 + 24 \sqrt 3 \sqrt{27 q^6 + q^9}\right)^{1\over 3}\,.
\end{equation}
Let
\begin{equation}
Y_{kj}=e^{2\pi i k\over 3}Y_j\,,\,k\,,j=1\,,2\,,3\,.
\end{equation}
Then we can get the coresponding $X_{kj}$ by equation \eqref{criticaldef}. From the expression of solutions in \eqref{ramificationpts}, we can see that all $X_{kj},\,Y_{kj}$ are invariant under some group action $\Gamma<\Gamma_0(3)$.

\begin{lem}\label{modularityofxycoordinate}
 Coordinates $X_{kj}$ and $Y_{kj}$ of ramification points are all modular functions of some group $\Gamma$.
\end{lem}

Let $u_r$ be the u-coordinate of ramification point $r\in R$. By the inverse map of uniformization, we can get the following lemma directly.
\begin{lem}\label{modularityofwpatcriticalpoints}
$\wp(u_r)$ and $\wp'(u_r)$ are meromorphic modular forms of group $\Gamma$ of weight 2 and 3 respectively.
\end{lem}
\begin{lem}
Let $r_1,\,r_2\in R$ be two ramification points, then $\wp(u_{r_1}-u_{r_2})$ and $\wp'(u_{r_1}-u_{r_2})$ are all modular forms. 
\end{lem}
\begin{proof}
Since we have relation
\begin{eqnarray}\label{grouplawforwp}
&&\wp(u_{r_1}-u_{r_2})={1\over 4}{\wp'(u_{r_1})+\wp'(u_{r_2})\over \wp(u_{r_1})-\wp(u_{r_2})}-\wp(u_{r_1})-\wp(u_{r_2})\,,\nonumber\\
&&\wp'(u_{r_1}-u_{r_2})={\wp'(u_{r_1})+\wp'(u_{r_2})\over \wp(u_{r_1})-\wp(u_{r_2})}\left(\wp(u_{r_1}-u_{r_2})-\wp(u_{r_1})\right)-\wp'(u_{r_1})\,.
\end{eqnarray}
Therefore by lemma \ref{modularityofwpatcriticalpoints}, we can get both $\wp(u_{r_1}-u_{r_2})$ and $\wp'(u_{r_1}-u_{r_2})$ are modular forms of weight 2 and 3 respectively.
\end{proof}

Involution can only be defined locally near ramification points in this example.
Assume $p=(X,Y)\in \ol{\mc{C}}_q$ near ramification point $(X_{kj},Y_{kj})$, then we have $p^{\ast}=(X,Y^{\ast})$, and
\begin{equation}\label{involutionexpression}
Y^{\ast}=e^{2\pi i k\over 3}\left(-Y^3-1-qX\right)^{1\over 3}\,.
\end{equation}

\section{Modularity of open-closed Gromov-Witten invariants}
\subsection{Bergman kernel}
We will use Schiffer kernel instead of Bergman kernel. To simplify notation, use $u_p$ instead of $u(p)$ and $\omega_{g,n}$ replacing $\omega_{g,n}^0$ since we fix framing $f=0$.
In the recursion, we need to expand $S(p,q)$ and $d^{-1}S(p,q)$ as power series of $u_q-u_r$, and take their coefficients, so we firstly need to study their coefficients carefully and prove that they are all almost meromorphic modular forms or almost meromorphic Jacobi forms.
\begin{equation}
S(p,q)=\left(\wp(u_p-u_q)+\hat \eta_1(\tau)\right)du_pdu_q\,.
\end{equation}
Near ramification point $r$,
\begin{equation}
S(u_p-u_q)=\sum_{k=0}^\infty S^{(k)}(u_p-u_r)(u_q-u_r)^k\,.
\end{equation}
$S^{(k)}(u_p-u_r)$ is obviously meromorphic Jacobi form.
Near ramification point $r\in R$, $u_{q^*}$ can be expressed as function of $Y$. 
\begin{equation}
u_{q^*}=\int_o^{Y^*}\omega\,.
\end{equation}
Then
\begin{align}
d^{-1}S(p,q):&=\int_{q^{*}}^q S(p,\eta)\,\nonumber\\
&=\left(\zeta(u_p-u_q)-\zeta(u_p-u_{q^{*}})+\hat{\eta}_1(u_q-u_{q^*})\right)du_p\,.
\end{align}
Expand $d^{-1}S(p,q)$ near ramification point $u_r$,
\begin{align}\label{expansionofB}
d^{-1}S(p,q)&=\left(S(u_p-u_r)(1-{\partial u(q^*)\over \partial u(q)}|_{q=r})\right)(u_q-u_r)\nonumber\\
&+ \left(-S'(u_p-u_r)+S'(u_p-u_r)\left({\partial u_{q^*}\over \partial u_q}\right)^2\right)-S(u_p-u_{q^*}){\partial^2u_{q^*}\over \partial u_q^2}(u_q-u_r)^2\nonumber\\
&+ O((u_q-u_r)^3)du_p\,.
\end{align}
In the above expansion,
\begin{equation}
{\partial u_{q^*}\over \partial u_q}={\partial u_{q^*}\over \partial Y_{q^*}}{\partial Y_{q^*}\over \partial Y_q}{\partial Y_{q}\over \partial u_q}|_{q=r}={\partial Y_{q^*}\over \partial Y_q}|_{q=r}\,.
\end{equation}
It is a meromorphic modular form when taking value at ramification points by Lemma \ref{modularityofxycoordinate}. Similarly, use chain rule to get ${\partial^n u_{q^*}\over \partial u_q^n}$, and we can prove it is also a meromorphic modular form when taking value at ramification points. Additionally, we have relation 
\begin{equation}
\wp'(u)^2=4\wp(u)^3-g_2(\tau)\wp(u)-g_3(\tau)\,,
\end{equation}
with $g_2\,,\,g_3$ given in \eqref{g2g3}. Therefore coefficient of $(u_p-u_r)^k$ in expansion of $d^{-1}S(p,q)$, denoted by $S_k$, is a polynomial of $\wp'$ and $\wp$ with almost meromorphic modular form as coefficients, so $S_k$ is almost-meromorphic Jacobi form. Additionally, in $S_k$, the order of differential of $\wp(u_p-u_r)$ is no more than $k-1$, so at ramification point $r\in R$, $S_k$ has order of pole no more than $k+1$.
Similarly for $S(u_p-u_r)$.
Collect these results as a lemma for later use,
\begin{lem}\label{modularirityofS}
Expanded near ramification point $r$, $d^{-1}S(p,q)=\sum_{k=0}^\infty S_k(u_q-u_r)^k$. Then the coefficient $S_k$ is an almost meromorphic Jacobi form as a polynomial of $\wp'(u_p-u_r)$ and $\wp(u_p-u_r)$. Further, $S_k$ only has poles at ramification point $r$, and the order of pole is no more than $k+1$.
Similarly, $S(u_p-u_{q^*})$ can be expanded as power series of $(u_q-u_r)$ whose coefficients are almost meromorphic Jacobi form and coefficient of $(u_q-u_r)^k$ has pole at $r$ of order no more than $k+2$.
\end{lem}

\subsection{$\lambda-\lambda^{\ast}$}
 $\lambda=\ln Y {dX\over X}|_{\mc{C}_q}$.
According to above subsection, we know
\begin{eqnarray}
&&Y^3-{Y^{\ast}}^3=2Y^3+1+qX\,,\nonumber\\
&&Y^3+{Y^{\ast}}^3=-1-qX\,.
\end{eqnarray}

Let $p_r=(x_r,y_r) \in R$.
\begin{align}
\lambda-\lambda^{\ast}&={1\over 3}\left(\ln\left(1+{Y^3-{Y^*}^3\over Y^3+{Y^*}^3}\right)-\ln\left(1-{Y^3-Y*^3\over Y^3+Y*^3}\right)\right){dX\over X}\nonumber
\\
&={2\over 3}\sum_{k=0}^\infty{1\over 2k+1}\left({2Y^3+1+qX\over -1-qX}\right)^{2k+1}{dX\over X}\nonumber\\
&={2Y^2\over 3X^2+qY^3}\sum_{k=0}^\infty{\left(2Y^3+1+qX\right)^{2k+2}\over (1+qX)^{2k+1} \cdot (2k+1)}{dY\over X}\,.
\end{align}
When expanded in Jacobi coordinate $u$,
\begin{equation}
(2Y^3+1+qX)^2= \left(-12 X_r^2+2q^2X_r+2q\right){\partial^2X\over \partial u^2}|_{u_r}\left(u-u_r\right)^2+O\left((u-u_r)^3\right)\,,
\end{equation}
with
\begin{equation}
-12 X_r^2+2q^2X_r+2q\neq 0\,,
\end{equation}
except for finite $q$ which are solutions of 
\begin{equation}
{q^{12}\over 864}+{q^9\over 12}+2q^6+17q^3+27=0\,.
\end{equation}
Above arguments give the following lemma,
\begin{lem}\label{modularityoflambda}
$\lambda-\lambda^{\ast}$ has two order of zeros at each ramification point, and when expanded as power series of $u-u_r$, all coeffecients of $\lambda-\lambda^{\ast}$ are meromorphic modular forms.
\end{lem}

\subsection{Modularity of Gromove-Witten invariants}
Apply Eynard-Orantin recursion, we can get $\widehat{\omega}_{g,n}$, and it has similar properties as in \cite{1805.04894}.
\begin{thm}\label{maintheorem}
\begin{enumerate}
\item
$\widehat{\omega}_{g,n}(u_1,\cdots,u_n)$ is symmetric in its argumengs, and only has poles at ramification points $R$. At each ramification point, the order of pole of each argument is no more than $6g+2n-4$, and total order of pole of all arguments is no more than $6g+4n-6$.
\item
$\widehat{\omega}_{g,n}(u_1,\cdots,u_n)$ is a polynomial of $S(u_{i}-u_r)$ and $S'(u_{i}-u_r)$ with coefficients lying in the ring of almost meromorphic modular forms, 
\begin{equation}
\widetilde{\mc{R}}=\mc{M}(\Gamma)\left[\{\wp(u_{r_i}-u_{r_j})\}_{r_i,r_j\in R},\hat{\eta}_1\right]\,.
\end{equation}
So, $\widehat{\omega}_{g,n}$ is an almost meromorphic Jacobi form in the differential ring
\begin{equation}
\widetilde{\mc J}(\Gamma(3)):=\widetilde{\mc{R}}\left[\{\wp(u_i-u_r),\wp'(u_i-u_r)\}_{i\in\{0,\cdots,n\},r\in R}\right]\,.
\end{equation}
\item For $g\geq 2$,
$\widehat{\mc{F}}_{g}={1\over 2g-2}\sum_{r\in R}Res_{q\to r}d^{-1}\lambda\cdot\omega_{g,1}$ is in the ring $\widetilde{\mc{R}}$.\
\end{enumerate}
\end{thm}

\begin{proof} All statemts are easy to prove by Lemma \ref{modularirityofS} and Lemma \ref{modularityoflambda} and are similar to the proof of Theorem $4.4$ in \cite{1805.04894}. I will give a skech. 
\item 1. As shown in below example \ref{omega03}, $\widehat{\omega}_{0,3}$ has pole of order 6 at ramification points totally, and each argument has order 2. In recursion \eqref{eorecursionrelation}, assume this statement is true for $\widehat{\omega}_{g_1,|I|+1},\,\widehat{\omega}_{g_2,|J|+1}$. Then at a ramification point $r$, argument $u_q$ has order of pole no more than $6g_1+2(|I|+1)-4$ in $\widehat{\omega}_{g_1,|I|+1}$, and no more than $6g_2+2(|J|+1)-4$ in $\widehat{\omega}_{g_2,|J|+1}$, so they totally will give $6g+2n-6$ order of pole. Since ${1\over 2}{d^{-1}S\over \lambda-\lambda^*}$ has simple pole at ramification point $r$, so coefficient of $(u_q-u_r)^{6g+2n-5}$ in expansion of $d^{-1}S$  will give the largest order of pole at $r$ after the residue, which has order of pole no more than $6g+2n-4$  by Lemma \ref{modularirityofS}. \\
Assume $\widehat{\omega}_{g_1,|I|+1}$ has total order of pole no more than $6g_1+4(|I|+1)-6$ and $\widehat{\omega}_{g_2,|J|+1}$ has total order of pole no more than $6g_2+4(|J|+1)-6$, so total order of pole of $\widehat{\omega}_{g,n}$ is no more than $6g_1+4(|I|+1)-6+6g_2+4(|J|+1)-6+2=6g+4n-6$.
\item 2. This statement is a direct result by Eynard-Orantin recursion and the expansion property of $d^{-1}S$ and $\lambda^*-\lambda$ proved in Lemma \ref{modularirityofS} and Lemma \ref{modularityoflambda}.
\end{proof}

I would like to give an example of $\widehat{\omega}_{0,3}$ to give a direct glance at $\widehat{\omega}_{g,n}$.
\begin{ex}\label{omega03}
\begin{align}
\widehat{\omega}_{0,3}(p_1,p_2,p_3)
&=\sum_{r\in R}Res_{q\to r}{{1\over2}d^{-1}S(p_1, q)\over \lambda-\lambda^*}\left(S(q,p_2)S(q^{*},p_3)+S(q^*,p_2)S(q,p_3)\right)\,\nonumber\\
&=\sum_{r\in R}{
\left(1-{\partial y^{*}\over \partial Y}|_{Y=Y_r}\right)S(u_{p_1}-u_r)S(u_{p_2}-u_r)S(u_{p_3}-u_r)\over
 {-2\over 3}\left({3Y_r^2\over \partial_X H(X,Y)|_{X_r}}{(-12X_r^2+2q^2X_r+2q)\over X_r(-1-qX_r)}{\partial^2X\over \partial u^2}{\partial Y\over \partial u}\right)|_{u=u_r}}
\end{align}
From the expression of $\omega_{0,3}(p_1,p_2,p_3)$, we know that, it is an almost meromorphic Jacobi form of weight 3. 
\end{ex}

Since $\widehat{\omega}_{g,n}$ is a polynomial of finite degree in terms of $\hat{\eta}_1$, we can take $Im\tau\to \infty$, then we will recover $\omega_{g,n}$.

\begin{thm}\label{omegaquasimodular}(modularity of generating functions)\\
\begin{enumerate}
\item. $d_{\tilde{X}_1}\cdots d_{\tilde {X}_n}\mc{F}_{(g,n)}^{\mc{X},\mc{L}}$ under mirror map is in the ring $\widehat{\mc{J}}(\Gamma)\otimes(H^*(\mc{B}\mu_3,\mbb{C}))^{\otimes n}$, where 
\begin{equation}
\widehat{\mc{J}}(\Gamma)=\widehat{\mc{R}}\left[\{\wp'(u_i-u_r),\wp(u_i-u_r)\}_{i\in\{1,\cdots,n\},r\in R}\right]\,,
\end{equation}
is the differential ring generated by $\wp'(u_i-u_r)$ and $\wp(u_i-u_r)$,
and 
 \begin{equation}
\widehat{\mc{R}}=\mc{M}(\Gamma)\left[\{\wp(u_{r_i}-u_{r_j})\}_{r_i.r_j\in R},\eta_1\right]\,,
\end{equation}
which is the ring of quasi-meromorphic modular forms.
\item. $\mc{F}_{g}^{\mc{X}}$ is quasi-meromorphic modular form of some group $\Gamma$ under closed mirror map. 
 Group $\Gamma$ depends on the one-dimensional subfamily we choose.
 \end{enumerate}
\end{thm}
\begin{proof}
From equation \eqref{mirrortheorem}, $d_{\tilde{X}_1}\cdots d_{\tilde {X}_n}\mc{F}_{(g,n)}^{\mc{X},\mc{L}}$ takes value in $\omega_{g,n}\otimes H^*(\mc{B}\mu_3,\mbb{C})$ and $\omega_{g,n}$ is quasi-meromorphic Jacobi form as limit of alomost meromorphic Jacobi form. 
\end{proof}

Before going to the modularity of open closed Gromov-Witten invariants, expand the $\omega_{g,n}$ in $X$ coordinate near open GW points.

In this special subfamily we are discussing now, the 3 open GW points are very easy, $\nu_0=(0,-1),\nu_1=(0,e^{\pi i\over 3}),\nu_2=(0,e^{2\pi i\over 3})$. Assume the corresponding $u$-coordinate of these points are $w_l,l=0,1,2$. Then From relation \eqref{inverseuniformization}, $\wp(w_l)$ and $\wp'(w_l)$ are modular forms. 
\begin{lem}\label{expansionatopengwpt}
Expand each argument of $\omega_{g,n}$ near an open GW point, $\omega_{g,n}$ can be written as a power series of $X$, with coefficients quasi-meromorphic modular forms.
\end{lem}
\begin{proof}
Near open GW point $w_l$,
\begin{equation}
\wp(u_i-u_r)=\wp(w_l-u_r)+\wp'(w_l-u_r){\partial u\over \partial X}|_{\nu_l}X+O\left(X^2\right)\,.
\end{equation}
Hence coefficients of $\wp(u_{i}-u_r)$ as power series of $X$ near the $l$-th open GW point are polynomial of $\wp(w_l-u_r),\,\wp'(w_l-u_r),\,{\partial^k u\over \partial X^k}|_{\nu_l}$. By \eqref{grouplawforwp} $\wp(w_l-u_r),\,\wp'(w_l-u_r)$ are modular forms. Therefore, $\omega_{g,n}\,$, as polynomial of $\wp(u_i-u_r)$ and $\wp'(u_i-u_r)$ by theorem \ref{omegaquasimodular}, can also be expanded as power series of $X$ with coefficients quasi-meromorphic modular forms.
\end{proof}

Fix degree of stable maps on n-boundaries  $((d_1,k_1),\cdots,(d_n,k_n))\in H_1(\mc{L},\mbb{Z})^n$, and consider open-closed GW invariants with this boundary degree  $\mc{N}^{d_1\cdots d_n}_{k_1\cdots k_n}$. Then we have it is quasi-meromorphic modular form.

\begin{thm} (Modularity of open GW invariants).\\
For $2g-2+n>0$,  open Gromov-Witten invariant $\mc{N}^{d_1\cdots d_n}_{k_1\cdots k_n}$
 is quasi-meromorphic modular form of some modular group $\Gamma$ under closed mirror map. 
\end{thm}
\begin{proof} 
According construction in \cite{fang2012open}, the open mirror map is given by 
\begin{equation}
\tilde X=A(q)X\,,
\end{equation}
Then we have,
\begin{equation}
\mc{N}^{d_1\cdots d_n}_{k_1\cdots k_n}={1\over A(q)^{\sum d_i}}\cdot {1\over \prod_{i=o}^nd_i!}{\partial^{d_1-1}\over \partial X_1^{d_1-1}}\cdots{\partial^{d_n-1}\over \partial X_n^{d_n-1}}|_{X_i=0}\sum_{l_1\cdots l_n \in {0,1,2}}{\xi_3^{\sum_{i=1}^nk_il_i}\over 3^n}{{(\rho^{l_1\cdots l_n})}^*\omega_{g,n}\over dX_1\cdots dX_n}\,.
\end{equation}
where $\xi_3$ is 3 order root of 1.
Since ${(\rho^{l_1\cdots l_n})}^*\omega_{g,n}$ is the expansion of $\omega_{g,n}$ near the open GW point $\nu_{l_i}$ for argument $X_i$, it is power series of $X_i$ with quasi-meromorphic modular forms as coefficients. Hence $\mc{N}^{d_1\cdots d_n}_{k_1\cdots k_n}$ is a quasi-meromorphic modular form. 
\end{proof}

\begin{rem}
This subfamily I choose now is just to simplify computation, the method can be hopefully extended to more general subfamily as follows. \\
1. Every cubic curve is birational to a Weierstrass normal form, hence we can get uniformization of a general cubic \eqref{toriccoordmirrorcurve}, and coordinate $X\,,\,Y$ are rational functions of $\wp,\,\wp'$ with coefficients in a filed which is a finite extension of field $\mbb{C}(a_1,\cdots,a_7)$. By a proper choice of one-dimensional subfamily, all $a_i$ are modular functions for some congruence subgroup $\Gamma\subset SL(2,\mbb{Z})$. See \cite{connell1996elliptic}.\\
2. The ramification points set $R=\{p\in \mc{C}, d(XY^f)=0\}$ is intersection of the following two plane curves.
\begin{align}\label{generalcubicramifpt}
&X+Y^3+1+a_1 {X^3\over Y^{3}}+a_2{X\over Y^{1}}+a_3{X^2\over Y^{2}}+a_4{X^2\over Y^{1}}+a_5XY+a_6Y^2+a_7Y=0\,\nonumber\\
&fX{-}3Y^3{+}(3f{+}3)a_1{X^3\over Y^3}{+}(f+1)a_2{X\over Y}{+}(2f{+}2)a_3{X^2\over Y^2}{+}a_4(2f{+}1){X^2\over Y}{+}(f{-}1)a_5XY{-}2a_6Y^2{-}a_7Y=0
\end{align}
Claim: the solutions of these two equations are modular functions for some smaller group $\Gamma'\subset \Gamma \subset SL(2,\mbb{Z})$.\\
3. Then we can apply the similar process in section 4 for specific subfamily, and get the similar results.
\end{rem}

\bibliographystyle{amsalpha}
\bibliography{referenceg1}

Department of Mathematics, University of Michigan, 2074 East Hall, 530 Church Street, Ann Arbor, MI 48109, USA
E-mail: yczhang15@pku.edu.cn

\end{document}